\documentclass[a4paper,11pt]{article}%
\pdfoutput=1 

\usepackage{graphicx}
\usepackage{amsthm}
\usepackage{amsmath}
\usepackage{amsfonts}
\usepackage{amssymb}
\usepackage{color}
\RequirePackage[colorlinks,citecolor=blue,urlcolor=blue]{hyperref}

\usepackage{multirow}
\usepackage{anysize}
\marginsize{3.4cm}{3.4cm}{3.2cm}{3.2cm}

\renewcommand{\P}{\mathbb{P}}

\newcommand{\R}{\mathbb{R}}

\newcommand{\E}{\operatorname{E}}
\renewcommand{\P}{\operatorname{P}}
\newtheorem{lm}{Lemma}

\newtheorem{prop}{Proposition}

\numberwithin{equation}{section}

\newtheorem{thm}{Theorem}[section]

\begin{document}

\begin{center}

{ \Large \bf An optimal investment strategy aimed at maximizing the expected utility across all intermediate capital levels}

\vspace{20pt}

{\bf J.~ Cerda-Hern\'andez}$\,^{a}$, {\bf A.~ Sikov}$\,^{a}$ and  {\bf A.~ Ramos}$\,^{b}$

\vspace{20pt}

{\footnotesize
$^a$~Econometric Modelling and Data Science Research Group, National University of Engineering, Lima-Peru.\\
$^b$~Department of Mathematics,
University of Tarapac\'a, Arica, Chile.\\
}

\vspace{30pt}

\end{center}

\begin{abstract}
This study investigates an optimal investment problem for an insurance company operating under the Cramer-Lundberg risk model, where investments are made in both a risky asset and a risk-free asset. In contrast to other literature that focuses on optimal investment and/or reinsurance strategies to maximize the expected utility of terminal wealth within a given time horizon, this work considers the expected value of utility accumulation across all intermediate capital levels of the insurer. By employing the Dynamic Programming Principle, we prove a verification theorem, in order to show that any solution to the Hamilton–Jacobi–Bellman  (HJB) equation solves our optimization problem. Subject to some regularity conditions on the solution of the HJB equation, we establish the existence of the optimal investment strategy.  Finally, to illustrate the applicability of the theoretical findings, we present  numerical examples.
\bigskip

{\small {\bfseries Keywords.} {Stochastic control $\cdot$  dynamic programming principle $\cdot$  Hamilton-Jacobi-Bellman equation $\cdot$  optimal investment $\cdot$  expected utility $\cdot$  risk process}}
\end{abstract}

\newpage
\section{Introduction}
The optimal investment problem has become an attractive research area in actuarial science, financial mathematics, and quantitative finance. For insurance companies, this problem involves determining the best way to allocate wealth among different asset classes, such as equity, debt and real estate, over a given period. The goal is to achieve this allocation but  also reducing risk, minimizing exposure to potential losses, and ensuring sufficient funds to meet future cash flow obligations.
\medskip

After the classical collective risk model introduced by  \cite{Lundberg1903}, the ruin probability of such a portfolio became a principal focus in this field, with various approaches. Nowadays, there are several additional methods for studying the wealth of an insurance portfolio and reducing or controlling the ruin probability. In stochastic control theory applied to insurance models, reinsurance or a combination of reinsurance with dividends and optimal investment are other possibilities for controlling certain risk measures of the insurance portfolio (see for instance  \cite{HipVog}, \cite{Asmusen_ruin2010}, \cite{SchmBook},   \cite{AzMuBook},  \cite{Diko}, \cite{Hipp2004},  \cite{Korn2008}, and references therein). 
\medskip

In this paper, we adopt the approach of maximizing the expected utility of an insurance portfolio with investment. In practical terms, most insurance contracts are inherently tied to financial markets, whether through mortality rates, interest rates, financial products, or direct connections to stocks or indices (see, for instance, \cite{Dhaene2013}, \cite{Artzner2023}, and \cite{Robben2022}). Therefore, it is essential to continually assess the risk and wealth of the insurance portfolio at all times $t$  in order to minimize exposure or risk. Additionally, numerous insurance companies calibrate their portfolios based on the severity of claims per policy or in response to an increase in the portfolio's risk, aiming to reduce fluctuations over the long term. To address these issues, unlike other literature that primarily focuses on optimal investment strategies to maximize the expected utility of terminal wealth within a specified time horizon (see, for instance, \cite{Cadenillas2014}, \cite{Badaoui2018}, \cite{Merton1969}), our proposal specifically considers the expected value of utility accumulation across all intermediate capital levels of the insurer.
\medskip

Following the ideas of  \cite{Cadenillas2014} and  \cite{Badaoui2018}, our model, described in Section~\ref{preliminaries}, proposes a new stochastic optimization problem. This problem differs from existing ones in the literature (see \cite{SchmBook} for a summary of stochastic control problems applied to insurance) by considering the expected value of utility accumulation across all intermediate capital levels of the insurer, because most insurance contracts are inherently tied to stocks or financial instruments, exposing the insurance portfolio to greater market risks \cite{Dhaene2013}, \cite{Artzner2023}.
\medskip

\cite{Ferguson1965} was the first to apply stochastic control theory to solve the problem of the expected utility of wealth for the investor in the discrete case. Subsequently, in his seminal paper, \cite{Merton1969} introduced the fundamental classical optimal investment-consumption model, laying the foundation for future continuous-time stochastic optimization problems.   Inspired by  \cite{Merton1969}, \cite{Browne1995} considered a risk process modeled by a Brownian motion with drift, incorporating the possibility of investment in a risky asset that follows a geometric Brownian motion, but without a risk-free interest rate. In that work, Browne verified the conjecture announced by Ferguson and, for the first time, established a relationship between minimizing the ruin probability and maximizing the exponential utility of terminal wealth. This connection provides a clear link between insurance and finance. Our work also establishes this relationship (see Lemmas \ref{prop:Vpi} and \ref{prop:Phipi}). Other related works that we can mention, where the problem of maximizing the expected utility of terminal wealth was studied, include  \cite{Zariphopoulou2001}, \cite{BF2013}. In several of these works, general properties were proven, along with an existence and uniqueness theorem, leading to a closed form of the optimal strategy.
\medskip

To solve our optimization problem, we explored the properties of the value function to derive the dynamic programming principle (DPP). Subsequently, we obtained the Hamilton-Jacobi-Bellman (HJB) equation. To find solutions to the HJB equation, we imposed additional conditions on the utility function and applied an appropriate boundary condition to the value function. Finally, we also demonstrate that in some cases, the Merton ratio serves as the optimal strategy, aligning with findings reported in \cite{Merton1969} and \cite{Browne1995}, among others.
\medskip

The rest of the paper is organized as follows. In Section \ref{preliminaries}, basic definitions are described and the optimization problem is formulated. Section \ref{value_function} focuses on studying the properties of the value function. In Section \ref{dynamic_programming}, the Dynamic Programming Principle for the problem is validated. In Section \ref{Hamilton-Jacobi},  it is demonstrated that  the value function is a solution to the HJB equation.  In Section \ref{numerical}, it is presented  numerical results in order to compare the behavior of the ruin probability with and without optimal investment, using exponential, Pareto, and Weibull claim size distributions. 

\section{Preliminaries and problem formulation}\label{preliminaries}

The well-known Cram\'er-Lundberg risk model, with application to insurance, is driven by equation
\begin{equation}\label{CL}
	X_t=x+ct - Q_t=x+ct - \sum_{i=1}^{N_t}U_i,\quad t\in[0,T]
\end{equation}
where $X_0=x$ is the initial surplus or the surplus known at a giving or starting instant, $X_t$ represents the dynamics of surplus of an insurance company up to time $t$, $c$ is the premium income per unit time, assumed deterministic and fixed, and $Q_t$ es the total  claim amount process. $\{U_i\}_{i\geq 1}$ is a sequence of i.i.d. random variables with common cumulative distribution  $F$, with $F(0)=0$. We assume the existence of $\mu= \mathbb{E}(U_1)$. $N_t=\max\{ k \geq 1:  T_k \leq t \}$ denote the number of claims occurring before or at a given time $t$, where the random variable $T_i$ denote the  arrival times. We assume that $\{N_t\}_{t\geq 0}$ is a Poisson process with Poisson intensity $\lambda$, independent of the sequence $\{U_i\}_{i\geq 1}$. The  process $Q_t= \sum_{i=1}^{N_t}U_i$ is then a compound Poisson process. We will assume particular distributions in some sections of this manuscript and we state it appropriate and clearly. An important condition for the model is the so called {\it income condition} or {\it net profit condition}, in the case positive loading condition: $c\E(T_1) > \E(U_1)$. It brings an economical sense to the model: it is expected that the income until the next claim is greater than the size of the next claim. If we denote $Y_1=T_1$ and $Y_i = T_i - T_{i-1}, i\geq 2$, then the net income between the $(i-1)$-th and  the $i$-th claims is $cY_i - U_i$
\medskip

Let the time to ruin of the Cram\'er-Lundberg process $X_t$  be denoted by $\tau=\inf \{t>0 : X_t <0 \}$. The ruin probability is defined as $\psi(x)=\P(\tau <\infty | X_0=x)$ and the corresponding  survival probability  is $\varphi(x)=1-\psi(x)$. 
\medskip

For $t\geq 0$ define $\mathcal{F}_t = \sigma\{X_s : 0\leq s\leq t \}$, the smallest $\sigma$-algebra making the family $\{X_t :  0\leq s\leq t \}$ measurable. $\{\mathcal{F}_t : t\geq 0\}$ is the filtration generated by $\{X_t\}$. Cramer-Lundberg  risk process $X_t$ has an independent increment property, i.e., for any $0\leq s \leq t$, the sub $\sigma$-algebra $\mathcal{F}_s$ and the random varaible $X_t-X_s$  are  independent.
\medskip

We shall assume that the insurance company  invests its surplus in a financial market described by the standard Black - Scholes model, i.e., on the market there is a riskless bond and risky assets satisfying the following SDEs 
\begin{equation}
	dS_t^0=rS_t^0dt\; ,  \qquad  dS_t = \mu S_tdt + \sigma S_tdB_t\;, \;\;\; t\in[0,T]
\end{equation}
respectively, where $r$ is the risk-free rate and $\mu,\sigma$ are the expected return and volatility of the stock market, and $dB_t$ is the increment of a standard Brownian motion. We assume that $\mu>r$.   At time $t$ the insurer must choose  what fraction of the surplus invest in a stock portfolio, $\pi_t$ (the remaining fraction $1-\pi_t$ being invested in the riskless bond),  i.e., $\pi_tX_t$ is the amount invested into the risky asset and  $(1-\pi_t)X_t$  is the amount invested into the riskless bond. We assume that the processes $S_t$ and $Q_t$ are independent, this means that the sub $\sigma$-algebras $\sigma\{S_t :  t\geq 0  \}$ and  $\sigma\{Q_t : t\geq 0 \}$ are independent. Denote by $X_t^{\pi}$ the process with  investment strategy $\pi=\{ \pi_t\}$ where $\pi_t\in[0,1]$. Given an investment strategy $\pi=\{ \pi_t\}$  it is easy to prove that the controlled risk process $X_t^{\pi}$  can be written as 
\begin{equation}\label{mertondyn}
	dX_t^{\pi}= \left[c + \mu\pi_tX_t^{\pi} + r(1-\pi_t)X_t^{\pi} \right] dt +  \sigma\pi_t X_t^{\pi}dW_t - dQ_t 
\end{equation}
The controlled risk process $X_t^{\pi}$  can be viewed  as the  wealth  of an risk averse economic agent (insurer) at time $t$ with cash injection to the fund due to the  income per unit time, assumed deterministic and constant. Note that, in this case, the wealth of the agent  have an stochastic dynamic with initial condition $X_0^{\pi}=x$. 
\medskip

We denote by  $\prod_{x}^{ad}$ the set of all the admissible investment strategies with initial value $x$. An investment strategy is {\it admissible} if the process $\{\pi_t\}$ is predictable with respect to the filtration $\mathcal{F}_t=\sigma\{X_t, S_t \}$.  We allow all adapted cadlag control processes $\pi=\{\pi_t\}\in \prod_{x}^{ad}$ such that  there is a unique solution $\{X_t^{\pi}\}$ to the stochastic differential equation (\ref{mertondyn}). 
\medskip

We define a stationary investment strategy as the one where the investment decision depends only on the current surplus, i.e., $\pi_t= \pi(X_t^\pi)$ is the fraction of the surplus invest in a stock portfolio when the current surplus is $X_t^\pi$. Thus the controlled investment  process $X_t^\pi$ should satisfy
\begin{equation}\label{mertondyn1}
	dX_t^{\pi}= [c + \mu \pi(X_{t^-}^{\pi}) + r( X_t^{\pi}   - \pi(X_{t^-}^{\pi})) ] dt +  \sigma \pi(X_{t^-}^{\pi})dW_t - dQ_t 
\end{equation}
If  $\pi: [0,\infty) \to \R$ is a Lipschitz continuous function,  using general results of diffusion processes is posible to prove that there is a unique strong solution to the SDE (\ref{mertondyn1}) (see \cite{Karatzas1991}, \cite{Watanabe1981}).  We further restrict to strategies such that  $X_t^{\pi}\geq 0$, i.e., the agent is not allowed to have debts. This means that if there is no money left the agent can no longer invest. We will further have to assume that our probability space $(\Omega,\mathcal{F},\mathbb{P})$ is chosen in such a way that for the optimal strategy found below a unique solution $\{ X_t^*\}$ exists.  Let us define the time to ruin $\tau^{\pi}=\inf\{ t : X_t^{\pi} < 0 \}$ and the ruin probability $\psi^{\pi}(x)=\mathbb{P}(\tau^{\pi}<\infty | X_0^\pi=x)$.   
\medskip

According to the utility theory, in a financial market where investors are facing uncertainty, an investor is not concerned with wealth maximization per se but with utility maximization. It is therefore possible to introduce an increasing and concave utility function $\phi(x,t)$ representing the expected utility of a risk averse investor. 
\medskip

We now describe our optimization problem. Given a strategy $\pi=\{\pi_t \}  \in \prod^{ad}$, we define the value of the strategy $\pi$ as follow
\begin{equation}
	V^{\pi}(x)= \mathbb{E}\left[ \int_0^{T\wedge\tau^\pi}  \phi(X_{s}^{\pi},s) ds \mid  X_0^{\pi}=x \right].
\end{equation}
Now, the goal of the problem is not anymore to maximize the expected portfolio value or minimize the ruin probability or maximizing the expectation of the present value of all dividends paid to the shareholders up to the ruin   but to maximize the expected utility stemming from the wealth during the contract $[0,T]$, where $T$ is the maturity date of the contract. If the  initial portfolio value is $X_0^{\pi}=x$, then our objective is to find the optimal strategy $\pi^*\in \prod_{x}^{ad}$ such that 
\begin{equation}\label{optimizationprob}
	V(x)= \sup_{\pi} V^{\pi}(x)= V^{\pi^*}(x)
\end{equation}
For simplicity of notation, we will omit $\pi$ of the stopping time $\tau^\pi$. We just make the convention that  $X_t^{\pi}=0$ for  $t\geq \tau$.
\medskip

Note that  this system have a two random components, the total claim amount $Q_t$ and the unit price $S_t$ of the risky asset, but the unique component that the  insurer can be to control is the fraction of wealth that is invested in the risky asset. 
\medskip

Because the agent prefer value growth  of the surplus $X_t^{\pi}$ (intuitively that property reduces ruin probability) the utility function $\phi(x,t)$   is assumed to be strictly increasing.  Because  the agent is risk-averse, the utility  $\phi(x,t)$  is assumed to be strictly concave. In other words, a monetary unit means less to the agent if $x$ is large than if $x$ is small. Finally, because the agent’s preferences do not change rapidly, we assume that $\phi(x,t)$   is continuous in $t$.  Note that as a concave function $\phi(x,t)$    is continuous in $x$. For simplicity of the notation we norm the utility functions such that  $\phi(0,t)=0$. To avoid some technical problems we suppose that $\phi(x,t)$  is continuously differentiable with respect to $x$ and that $\lim_{x\to \infty} \phi_x(x,t)=0$, where  $\phi_x(x,t)$  denotes the derivative with respect to $x$.

\section{Basic properties of the value function}\label{value_function}
In this section, we present some results that characterize the regularity of the value function $V$ of our optimization problem.
\begin{lm}\label{lemmaconcave}
	The function $V$ is strictly increasing and concave with boundary value $V(0)=0$, and hence continuous in the interior of the domain.
\end{lm}
\begin{proof}
	Note that if $x=0$, any strategy with $\pi_t \neq 0$ would immediately lead to ruin. Without wealth the economic agent don't obtain any utility. Thus, $V(0)= \int_0^T\phi(0,s)ds = 0$. Let $\pi$ be a strategy for initial capital $x$ and let $y>x$. We denote the surplus process starting in $x$ by $X^\pi$ and the surplus process starting in $y$ by $Y^\pi$. Note that $X_t^\pi < Y_t^\pi$, for all $t\geq 0$. As  $\phi$ is strictly increasing, then  $\phi(X_t^\pi,t) < \phi(Y_t^\pi,t)$. Thus, we obtain $V(y)>V(x)$.
	
	Let $z=\alpha x +  (1-\alpha)y$ for $\alpha\in[0,1]$. Let $\pi^1$ the strategy  for initial value $x$ and $\pi^2$ the strategy for initial value $y$.  Consider the strategy $\pi_t = \alpha \pi_t^1  + (1-\alpha)\pi_t^2$. Then  $X_t^\pi = \alpha X_t^{\pi^1} + (1-\alpha) X_t^{\pi^2}$. The value of this new strategy now becomes
	$$
	\begin{array}{ccl}
		V(z) &\geq& V^\pi(z)=\E \left[ \displaystyle\int_0^{T\wedge \tau^{\pi}}\phi\left(\alpha X_s^{\pi^1} + (1-\alpha) X_s^{\pi^2}, s\right)ds \right] \vspace{0.2cm} \\
		&\geq& \E \left[ \displaystyle\int_0^{T\wedge \tau^\pi} \alpha  \phi\left(X_s^{\pi^1},s\right)  + (1-\alpha)\phi\left( X_s^{\pi^2}, s\right)ds \right]  \vspace{0.2cm} \\
		&=& \alpha V^{\pi^1}(x) + (1-\alpha)V^{\pi^2}(y) 
	\end{array}
	$$
	Taking the supremum on the right-hand side we obtain $V(z)\geq \alpha V(x) + (1-\alpha)V(y)$. 
\end{proof}

The following lemma establishes that for any strategy $\pi\in\prod_{x}^{ad}$,  such that  $\tau^{\pi}$  occurs between claim times, can never be optimal. 

\begin{lm}
	Suppose that  $\pi\in \prod_{x}^{ad}$ is such that  $\mathbb{P}(T_i\wedge T< \tau^\pi < T_{i+1}\wedge T) > 0$, for some $i\in\mathbb{N}$, where $T_i$'s are the jump times of the Poisson process $\{N_t : t\geq 0 \}$, then exist $\widetilde{\pi}\in\prod_{x}^{ad}$ such that  $\mathbb{P}\left( \tau^{\widetilde{\pi}} =  T_i \;\; \mbox{for some} \;  i \right)=1$, and  $V^{\widetilde{\pi}}(x) > V^{\pi}(x)$.
\end{lm}
\begin{proof} Firts note that on the set  $\{ T_i\wedge T< \tau^\pi < T_{i+1}\wedge T \}$, one must have that  $X_{\tau^{\pi -} }^\pi = X_{\tau^{\pi} }^\pi = 0$. Now,  define the strategy  $\widetilde{\pi}_t = \pi_t 1_{\{t<\tau^\pi \} }$ and denote  $\widetilde{X}=X^{\widetilde{\pi}}$.  Then, it is clearly that  $\widetilde{X}_t = X_t^{\pi}$ for all $t\in [0,\tau^\pi]$,  $\mathbb{P}$-a.s. Consequently, $\widetilde{X}$ satisfies the following  stochastic diferential equation
	\begin{equation}
		d\widetilde{X}_t = \left\{
		\begin{array}{lcl}
			dX_t^\pi &,& t\in[0,\tau^\pi] \vspace{0.2cm}\\
			dX_t=(c+rX_t)dt - dQ_t &,& t> \tau^\pi
		\end{array}\right.
	\end{equation}
	Note that the process  $dX_t=(c+rX_t)dt - dQ_t$ in the absence of claims (or  between the jumps of $N_t$) is increasing, and  the ruin can occur only at some time $T_i$. Therefore the $\tau^{\widetilde{\pi}} =  T_k$  for some  $k>i$. Thus
	$$
	\begin{array}{rcl}
		V^{\widetilde{\pi}}(x) &=& \mathbb{E}\left[ \displaystyle\int_0^{T\wedge\tau^{\widetilde{\pi}}}  \phi(X_{s}^{\widetilde{\pi}},s) ds |  X_0^{\widetilde{\pi}}=x \right]     \vspace{0.2cm}\\
		&= &   \mathbb{E}\left[ \displaystyle\int_0^{T\wedge\tau^{\pi} }  \phi(X_{s}^{\widetilde{\pi}},s) ds |  X_0^{\widetilde{\pi}}=x \right]   +    \mathbb{E}\left[ \displaystyle\int_{T\wedge\tau^{\pi} }^{T\wedge\tau^{\widetilde{\pi}}}  \phi(X_{s}^{\widetilde{\pi}},s) ds |  X_0^{\widetilde{\pi}}=x \right]    \vspace{0.2cm}\\ 
		&=&   V^{\pi}(x)   +    \mathbb{E}\left[ \displaystyle\int_{\tau^{\pi} }^{T\wedge\tau^{\widetilde{\pi}}}  \phi(X_{s}^{\widetilde{\pi}},s) ds |  X_0^{\widetilde{\pi}}=x \right]    \vspace{0.2cm}\\ 
		& \geq & V^{\pi}(x)   +    \mathbb{E}\left[  \displaystyle\int_{\tau^{\pi} }^{T_{i+1}\wedge T } 1_{ \{ T_i\wedge T< \tau^\pi < T_{i+1}\wedge T \} }  \phi(X_{s}^{\widetilde{\pi}},s) ds    |  X_0^{\widetilde{\pi}}=x \right]    \vspace{0.2cm}\\ 
		&>&V^{\pi}(x) 
	\end{array}
	$$
	since $\mathbb{P}(T_i\wedge T< \tau^\pi < T_{i+1}\wedge T) > 0$.  This prove the lemma.
\end{proof}

Using the before lemma we obtain the following result about the ruin probability. 

\begin{lm}\label{prop:Phipi}
	Suppose that  $\pi\in \prod_{x}^{ad}$ is such that  $\mathbb{P}\left[ T_i\wedge T< \tau^\pi < T_{i+1}\wedge T\right] > 0$, for some $i\in\mathbb{N}$, where $T_i$'s are the jump times of the Poisson process $\{N_t : t\geq 0 \}$.  Then, exists $\widetilde{\pi}\in\prod^{ad}$ such that  $\mathbb{P}_{x}\left( \tau^{\widetilde{\pi}} =  T_i \;\; \mbox{for some} \;  i \right)=1$, and  
	$$\psi^{\widetilde{\pi}}(x)  <  \psi^{\pi}(x)$$
\end{lm}
\begin{proof} On the event  $\{ T_i\wedge T< \tau^\pi < T_{i+1}\wedge T \}$,  with the notation as before lemma,  we have that  $X_{\tau^{\pi -} }^\pi = X_{\tau^{\pi} }^\pi = 0$. Now,  define the strategy  $\widetilde{\pi}_t = \pi_t 1_{\{t<\tau^\pi \} }$ and define  by  $\widetilde{X}$ the risk process  with investment, i.e., $\widetilde{X}:=X^{\widetilde{\pi}}$.  Then, it is clearly that  $\widetilde{X}_t = X_t^{\pi}$ for all $t\in [0,\tau^\pi]$,  $\mathbb{P}$-a.s.  Consequently, $\widetilde{X}$ satisfies the following  stochastic diferential equation
	\begin{equation}
		d\widetilde{X}_t = \left\{
		\begin{array}{lcl}
			dX_t^\pi &,& t\in[0,\tau^\pi] \vspace{0.2cm}\\
			dX_t=(c+rX_t)dt - dQ_t &,& t> \tau^\pi
		\end{array}\right.
	\end{equation}
	Note that the process  $dX_t=(c+rX_t)dt - dQ_t$ in the absence of claims  is increasing, and  the ruin can occur only at some time $T_i$. Therefore the $\tau^{\widetilde{\pi}} =  T_k$  for some  $k>i$.  Thus, 
	$$\tau^{\widetilde{\pi}} = \tau^{\pi} + \theta$$
	where $\theta$ is the ruin time of the process  $\{ Y_t=\widetilde{X}_{t + \tau^{\pi}} : t\geq 0 \}$, with  $Y_0=0$. Since  $\mathbb{P}(T_i\wedge T< \tau^\pi < T_{i+1}\wedge T) > 0$, we have that $\mathbb{P}(\theta > 0)>0$. Thus,
	$$
	\begin{array}{rcl}
		\psi^{\widetilde{\pi}}(x)  &=&  \mathbb{P}\left[  \tau^{\widetilde{\pi}} < \infty |  X_0^{\widetilde{\pi}}=x  \right] = \mathbb{P} \left[  \tau^{\pi} < \infty \; \wedge  \;  \theta  < \infty |  X_0^{\pi}=x  \right]   \vspace{0.2cm}\\
		&=&  \mathbb{P} \left[  \tau^{\pi} < \infty|  X_0^{\pi}=x  \right]  \mathbb{P}_x \left[   \theta  < \infty |  \tau^{\pi} < \infty   \right]   \vspace{0.2cm}\\
		&=&  \psi^{\pi}(x) \cdot   \mathbb{P}_x \left[   \theta  < \infty |  \tau^{\pi} < \infty   \right]   \vspace{0.2cm}\\
	\end{array}
	$$
	Remember that the initial surplus of  the process $\{Y_t : t\geq 0 \}$ is  $Y_0=0$, and the fraction of the surplus invested  in the risky  asset and the  riskless bond is zero and one, respectively, i.e., $\pi_t=0$ for all $t\geq 0$. Denote by $\psi^{f}(x)$ the ruin probability of the process $Y_t$. It is clear that $\psi^{f}(x) < \psi (x)$, where  $\psi (x)$ is the classical Cram\'er-Lundberg risk process with the same arrival times, starting  at $\tau^\pi$, and the same claim size. Since we assume  the net profit condition to the process without  investment, we obtain  $\psi (x)<1$. Therefore, we have $ \mathbb{P}_x \left[   \theta  < \infty |  \tau^{\pi} < \infty   \right]<1$,  proving the  lemma.
\end{proof}

\begin{lm}\label{prop:Vpi}
	For any $\varepsilon >0$, there exist $\delta>0$ such that  for any  strategy $\pi\in \prod_{x}^{ad}$ and $h>0$ with $0<h<\delta$, we can find  $\hat{\pi}^h\in \prod_{x}^{ad}$ such that
	\begin{equation}\label{ucpi}
		V^{\pi}(x) - V^{\hat{\pi}^h}(x-h) < \varepsilon
	\end{equation}	
\end{lm}
\begin{proof}
	In order to solve the problem, we introduce the value function
	\begin{equation}
		V^{\pi}(t,x)= \mathbb{E}\left[ \int_t^{\tau^\pi\wedge T
		}  \phi(X_{s}^{\pi},s) ds \mid  X_t^{\pi}=x \right].
	\end{equation}
	The value function becomes  $V(x,t)= \sup_{\pi} V^\pi (x,t)$. The function we are looking for is  $V(x)=V(x,0)$.  Now, we consider the following modified strategy 
	\begin{equation}
		\hat{\pi}_t =  1_{\{T_1 >h \} } 1_{[h, \tau^\pi\wedge T] }(t) \pi_t
	\end{equation}
	where  $T_1$  is the first interarrival time. If  $t<h$ and on the set $\{T_1 > h \}$ we have that   $\hat{\pi}_t \equiv 0$ and  the solution of   $d X_t^{\hat{\pi}} = (c + rX_t^{\hat{\pi}})dt$ will be continuous and increasing for $t\in [0,h]$, so that 
	$$
	\begin{array}{ccl}
		X_t^{\hat{\pi}} &=& xe^{rt}  + \dfrac{c}{r} \left( e^{rt} -1 \right), \; t\in [0,h] \;\;\; \mathbb{P}-\mbox{a.s.  on}\;  \{T_1 > h \}
	\end{array}
	$$
	For fixed $\varepsilon >0$, we have
	$$V^{\pi}(x) - V^{\hat{\pi}^h}(x-h)= \left( V^{\pi}(x) - V^{\hat{\pi}^h}(x) \right) + \left( V^{\hat{\pi}^h}(x) - V^{\hat{\pi}^h}(x-h) \right):= J_1 + J_2$$
	We shall estimate $J_i$'s  separately. First,  since $V(t,x)$ is increasing, concave and continuous in $[0,\infty )$ with $\lim_{x\to\infty}V(t,x)<\infty$, we have that $V(t,x)$ is uniformly continuous on the variables $(t,x)$, for all  strategy  $\pi \in \prod^{ad}$. Thus, there exists $\delta_1$ such that, for $h< \delta_1$, then  $J_2 < \varepsilon/2$. 
	
	Now, to estimate  $J_1$,  we use the follow inequality
	$$
	\begin{array}{ccl}
		V^{\hat{\pi}}(0,x) &=&  \mathbb{E}\left[ \displaystyle\int_0^{\tau^{\hat{\pi}} \wedge T
		}  \phi(X_{s}^{\hat{\pi}},s) ds \mid  X_0^{\hat{\pi}}=x \right]   \vspace{0.2cm}  \\
		&\geq &  \mathbb{E}_x\left[ \displaystyle\int_0^{\tau^{\hat{\pi}} \wedge T
		}  \phi(X_{s}^{\hat{\pi}},s) ds \mid  T_1 > h \right] \mathbb{P}(T_1 > h)   \vspace{0.2cm} \\
		&\geq & e^{-\lambda h} \mathbb{E}_x\left[ \displaystyle\int_0^{\tau^{\hat{\pi}} \wedge T
		}  \phi(X_{s}^{\hat{\pi}},s) ds \mid  T_1 > h \right]   \vspace{0.2cm} \\
		&\geq &  e^{-\lambda h} \mathbb{E}_x V^{\pi}(h,x)
	\end{array}
	$$
	Thus, using the  previous argument and  since $V$ is increasing, concave and continuous in $[0,\infty )$ with $\lim_{x\to\infty}V(x)<\infty$, we have that
	\begin{equation}
		V^{\pi}(h,x)  - V^{\hat{\pi}}(x)  \leq  \left( 1-e^{-\lambda h} \right) V^{\pi}(h,x) \leq C_1h  \\
	\end{equation}
	In addition, it is clear that $V(t,x)$ is uniformly continuous on the variables $(t,x)$  with  $\lim_{h\downarrow 0}[V^{\pi}(h,x) - V^{\pi}(0,x)]=0$, for all  strategy  $\pi \in \prod^{ad}$. Therefore, there exists $\delta_2$ such that $J_1< \varepsilon/2$ for   $h<\delta_2$. Taking  $\delta = \min\{\delta_1 , \delta_2 \}$  we prove (\ref{ucpi}), whence the lemma.
\end{proof}

\section{Dynamic programming principle}\label{dynamic_programming}

In this section, we show the  Dynamic Programming Principle for  our optimization problem. 

\begin{thm}\label{DPP}
	For any initial capital $x$ and stopping time $\tau\in[0, T]$,  the value function satisfies
	\begin{equation}\label{dpptheo}
		V(x)= \sup_{\pi\in \prod^{ad}} \mathbb{E}_x \left[  \int_{0}^{\tau\wedge \tau^\pi} \phi( X_s^\pi, s ) ds  + V(X^{\pi}_{\tau\wedge \tau^\pi}) \right]
	\end{equation}
\end{thm}

\begin{proof}
	We shall first argue the theorem for deterministic time $h\in[0,T]$.  Define  the function 
	$$A(x,h)= \sup_{\pi\in \prod^{ad}} \mathbb{E}_x \left[  \int_{0}^{h \wedge \tau^\pi} \phi( X_s^\pi, s ) ds  + V(X^{\pi}_{h \wedge \tau^\pi}) \right]$$ 
	Now, we show that  $V(x)= A(x,h)$. Let $\pi\in \prod^{ad}$, and write
	\begin{equation}\label{eqn:Vpi}
		V^\pi(x)=  \mathbb{E}_x\left[ \int_0^{h \wedge\tau^\pi}   \phi(X_{s}^{\pi},s) ds\right] +   \mathbb{E}_x\left[ \int_h^{\tau^\pi}  1_{\{ \tau^\pi >h\} }  \phi(X_{s}^{\pi},s) ds \right]
	\end{equation}	
	and 
	$$
	\begin{array}{l}
		\mathbb{E}_x\left[ \displaystyle\int_h^{\tau^\pi}  1_{\{ \tau^\pi >h\} }  \phi(X_{s}^{\pi},s) ds  \right] \vspace{0.2cm}  \\
		= \mathbb{E}_x \left[   1_{\{ \tau^\pi  >h\} }  \mathbb{E}\left[ \displaystyle\int_h^{\tau^\pi}  \phi(X_{s}^{\pi},s) ds  \mid  \mathcal{F}_h \right]   \right] \vspace{0.2cm}  \\ 
		\leq  \mathbb{E}_x \left[  1_{\{ \tau^\pi >h\} }  V(X_{h}^\pi)   \right] \vspace{0.2cm} \\
		\leq  \mathbb{E}_x \left[  V(X_{h\wedge \tau^\pi}^\pi) \right] 
	\end{array}
	$$
	Then
	$$
	V^\pi(x) \leq  \mathbb{E}_x\left[ \int_0^{h \wedge\tau^\pi}  \phi(X_{s}^{\pi},s) ds\right] +  \mathbb{E}_x \left[  V(X_{h\wedge \tau^\pi}^\pi) \right] 
	$$
	Taking supremum over strategies  $\pi$, we obtain $V(x)\leq A(x,h)$.
	\medskip
	
	The reverse inequality is a bit more complicated. Let $\varepsilon >0$ be a constant.   Since $V$ is increasing, concave and continuous in $[0,\infty )$ with $\lim_{x\to\infty}V(x)<\infty$, we can find an increasing sequence $\{x_n : n=0,1, 2, ...\}$ with $x_0=0$ and $\lim_{n\to\infty}x_n=\infty$, such that, if $x \in [x_n, x_{n+1})$ then
	\begin{equation}\label{dpp:eqn1}
		V(x) - V(x_{n}) < \dfrac{\varepsilon}{3}
	\end{equation}
	Now, by definition of  function $V$, for each   $x_i$  there exists an admissible  strategy $\pi^{(i)}  \in \prod^{ad}$ such that  
	\begin{equation}\label{dpp:eqn2}
		V(x_i) - V^{\pi^{(i)}}(x_i) < \dfrac{\varepsilon}{3}
	\end{equation}
	If $x \in [x_i, x_{i+1})$ then there exists an admissible  strategy $\widetilde{\pi}^{(i)}  \in \prod^{ad}$ such that  
	\begin{equation}\label{dpp:eqn3}
		V^{\pi^{(i)}}(x_i)  -  V^{\widetilde{\pi}^{(i)}}(x_i)   < \dfrac{\varepsilon}{3}
	\end{equation}
	Then, we obtain the following inequalities
	\begin{equation}\label{dpp:eqn4}
		\begin{array}{ccl}
			V^{\widetilde{\pi}^{(i)}}(x) &>& V^{\widetilde{\pi}^{(i)}}(x_i) -  \dfrac{\varepsilon}{3} \vspace{0.2cm} \\ 
			& > & V(x_i) -  \dfrac{2\varepsilon}{3}  \vspace{0.2cm} \\ 
			& > & V(x) - \varepsilon
		\end{array}
	\end{equation}
	Therefore, we get that $V(x) - V^{\widetilde{\pi}^{(i)} }(x)  < \varepsilon$. Now, for any $\pi \in \prod^{ad}$  we will define a new strategy  $\pi^*$ as follow:
	\begin{equation}\label{dpp:eqn5}
		\pi_t^* = \pi_t 1_{[0,h)}(t) + \sum_{i=0}^{\infty} 1_{\left[h, \tau^\pi \right]}(t) 1_{A_i}(X_h^\pi) \widetilde{\pi}^{(i)} 
	\end{equation}
	\begin{figure}[t] 
		\centering
		\includegraphics[width=10cm]{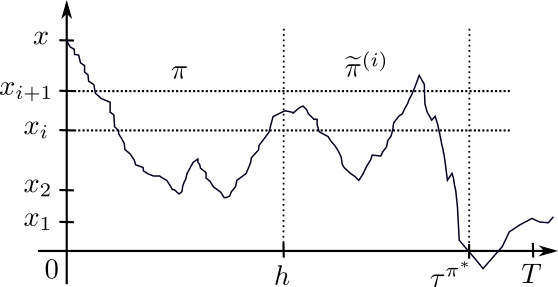}
		\caption{Graphical representation of the  strategy $\pi^*$ defined in (\ref{dpp:eqn5}).}
		\label{strategyDPP}
	\end{figure}
	where  $A_i=[x_i , x_{i+1} ) $ and  $h>0$.   Notice that,  for  $t \leq h$ we  have that    $\pi^{*}_t = \pi_t$. In addition,   in the case   $h< \tau^\pi$, we take the index  $i$  such that  $X_{h}^\pi=X_{h\wedge \tau^\pi}^\pi \in A_{i}$,  and follow the strategy  $\pi^{*}_{t}=\pi_t^{(i)}$, for  $t\in [h, \tau^\pi]$.  Figure \ref{strategyDPP} shows the  graphical representation of the  strategy $\pi^*$.  In addition, notice that if $h< \tau^\pi$, then  $X^\pi_t > 0$, for all $t\in[0,h )$,  $\pi^{*}_t = \pi_t$   for $t\in[0,h\rangle$ and  $\pi^{*}_t = \pi_t^{(i)}$ for  $t\in[h, \tau^\pi]$. Therefore, we have that  
	$\{ \tau^{\pi^*}\leq h  \} = \{ \tau^{\pi}\leq h \} $.  Furthermore, on $\{ \tau^{\pi^*} > h \}$, by Equation (\ref{dpp:eqn4}) we obtain 
	\begin{equation}\label{dpp:eqn6}
		V^{\pi^*}(X_h^{\pi}) \geq  V(X_h^{\pi}) - \varepsilon, \quad \mathbb{P}-\mbox{a.s.  on } \;  \{ \tau^{\pi^*} > h \}
	\end{equation}
	Consenquently,  similar to  Equation (\ref{eqn:Vpi}) we have that 
	
	$$
	\begin{array}{ccl}
		V(x) \geq V^{\pi^*}(x) & = &  \mathbb{E}_x\left[ \displaystyle \int_0^{h \wedge\tau^{\pi^*}}   \phi(X_{s}^{\pi^*},s) ds\right] +   \mathbb{E}_x\left[ \displaystyle \int_h^{\tau^{\pi^*}}  1_{\{ \tau^{\pi^*}  >h\} }  \phi(X_{s}^{\pi^*},s) ds \right]  \vspace{0.2cm} \\
		& = &   \mathbb{E}_x\left[ \displaystyle\int_0^{h \wedge\tau^{\pi}}   \phi(X_{s}^{\pi},s) ds\right] +   \mathbb{E}_x\left[ 1_{\{ \tau^{\pi^*}  >h\} }    \displaystyle \int_h^{\tau^{\pi^*}}  \phi(X_{s}^{\pi^*},s) ds \right]  \vspace{0.2cm} \\
		& =  &   \mathbb{E}_x\left[ \displaystyle\int_0^{h \wedge\tau^{\pi}}   \phi(X_{s}^{\pi},s) ds\right] +  \mathbb{E} \left[ 1_{\{ \tau^{\pi^*}  >h\} }    \displaystyle \int_h^{\tau^{\pi^*}}  \phi(X_{s}^{\pi^*},s) ds  \mid X^{\pi^*}=x \right]  \vspace{0.2cm} \\
		& =  &   \mathbb{E}_x\left[ \displaystyle\int_0^{h \wedge\tau^{\pi}}   \phi(X_{s}^{\pi},s) ds\right] +  \mathbb{E} \left[ 1_{\{ \tau^{\pi^*}  >h\} } \mathbb{E}   \left[      \displaystyle \int_h^{\tau^{\pi^*}}  \phi(X_{s}^{\pi^*},s) ds \mid  X_h^{\pi^*}  \right]  \mid X_0^{\pi^*}=x \right]  \vspace{0.2cm} \\
		& =  &   \mathbb{E}_x\left[ \displaystyle\int_0^{h \wedge\tau^{\pi}}   \phi(X_{s}^{\pi},s) ds\right] +  \mathbb{E}_x \left[ 1_{\{ \tau^{\pi^*}  >h\} } V^{\pi^*}(X_{h \wedge\tau^{\pi}}^{\pi})   \right] 
	\end{array}
	$$
	Now, we use the fact that  
	$$1_{\{ \tau^{\pi^*}  \leq  h\}}  V^{\pi^*} (X_{h \wedge\tau^{\pi}}^\pi)  
	=  1_{\{ \tau^{\pi}  \leq  h\}}  V^{\pi^*} (X_{\tau^{\pi}}^\pi) = 1_{\{ \tau^{\pi^*}  \leq  h\}}  V^{\pi^*} (0))=0$$ 
	Therefore, using Equation (\ref{dpp:eqn6}), we get 
	\begin{equation}
		\begin{array}{ccl}
			V(x) \geq V^{\pi^*}(x) & = &  \mathbb{E}_x\left[ \displaystyle\int_0^{h \wedge\tau^{\pi}}   \phi(X_{s}^{\pi},s) ds  +    V^{\pi^*}(X_{h \wedge\tau^{\pi}}^{\pi})   \right] \vspace{0.2cm} \\ 
			&\geq &   \mathbb{E}_x\left[ \displaystyle\int_0^{h \wedge\tau^{\pi}}   \phi(X_{s}^{\pi},s) ds  +    V^{\pi}(X_{h \wedge\tau^{\pi}}^{\pi})   \right]  - \varepsilon \vspace{0.2cm} \\ 
			&\geq & A(x,h) - \varepsilon 
		\end{array}
	\end{equation}
	Since $\varepsilon >0$ is arbitrary, we obtain that  $V(x)\geq A(x,h)$, proving  the theorem for  $\tau=h$.
	\medskip
	
	We now consider the general case, when  $\tau\in [0,T]$. Let  $t_0=0, t_1,..., t_n=T$ be a partition of $[0,T]$, where  $t_k=\dfrac{k}{n}T$,  $k=0,1,...,n$.  Define $$\tau_n=\sum_{k=0}^{n-1}t_k 1_{[t_k, t_{k+1})} (\tau) = \left\{
	\begin{array}{ccl}
		t_0\equiv 0 & \mbox{if} & \tau\in[t_0, t_1) \\
		t_1 & \mbox{if} & \tau\in[t_1, t_2) \\
		\vdots  &  & \vdots   \\
		t_{n-1} & \mbox{if} & \tau\in[t_{n-1}, t_n) 
	\end{array}
	\right.$$ 
	Note that $\{\tau_n\}$  are simple functions such that  $\tau_n \rightarrow \tau $, $\mathbb{P}$-a.s.  Further, $\tau_n$  take only a finite number of values. Further,  $\tau_n\equiv 0$ if $\tau\in [0, t_1)$ and   $\tau_n \geq t_1$ if $\tau\in [t_1,T)$,   for all $n\geq 1$.  Using the same argument as (\ref{eqn:Vpi}),  it is easy to show that  $V(x) \leq A(x,\tau_n) \;\; \mbox{for all} \; n\geq 1$.
	\medskip 
	
	\noindent  The reverse inequality  shall prove utilizing induction on $n$, in order to 
	\begin{equation}\label{eqn:induction}
		V(x) \geq A(x,\tau_n) \;\; \mbox{for all} \; n\geq 1
	\end{equation} 
	For  $n=1$ we have that $\tau_1 \equiv 0$ on $[0,T)$, so there is nothing to prove. Suppose that  (\ref{eqn:induction}) holds for  $\tau_{n}$. We shall argue  that (\ref{eqn:induction})  holds for $\tau_{n+1}$ as well. For any $\pi\in\Pi^{ad}$, on  $\{\tau_{n+1} < t_1 \}$ we have  that $\tau_{n+1}\equiv 0$ and  
	$$\mathbb{E}_x \left[ 1_{\{\tau_{n+1} < t_1 \}} \left( \displaystyle \int_{0}^{\tau_{n+1} \wedge \tau^\pi} \phi( X_s^\pi, s ) ds  + V(X^{\pi}_{\tau_{n+1} \wedge \tau^\pi}) \right) \right] = \mathbb{E}_x \left[ V(X^{\pi}_{0} ) \right ]$$
	Then, we have  $V(x,\tau_{n+1})=V(x)$ $\mathbb{P}$-a.s. on $\{\tau_{n+1} < t_1 \}$  for all $n$.  Now, on $\{\tau_{n+1} \geq t_1 \}$ we have
	
	\begin{equation}\label{dpp:eqn7}
		\begin{array}{l}
			\mathbb{E}_x \left[  \displaystyle \int_{0}^{\tau_{n+1} \wedge \tau^\pi} \phi( X_s^\pi, s ) ds  + V(X^{\pi}_{\tau_{n+1} \wedge \tau^\pi}) \right] \vspace{0.2cm} \\
			=	\mathbb{E}_x \left[ 1_{ \{\tau^\pi < t_1 \} }  \displaystyle \int_{0}^{\tau^\pi} \phi( X_s^\pi, s ) ds    \right] \vspace{0.2cm} \\
			+  \; 	\mathbb{E}_x \left(  \left[  \displaystyle \int_{0}^{\tau_{n+1} \wedge \tau^\pi} \phi( X_s^\pi, s ) ds  + V(X^{\pi}_{\tau_{n+1} \wedge \tau^\pi}) \right]  1_{ \{\tau_{n+1}  > t_1 \}} 1_{ \{\tau^\pi \geq t_1 \} }     \right. \vspace{0.2cm} \\
			+  \; 	\left.  \left[  \displaystyle \int_{0}^{t_1} \phi( X_s^\pi, s ) ds  + V(X^{\pi}_{t_1}) \right]  1_{ \{\tau_{n+1}  = t_1 \}} 1_{ \{\tau^\pi \geq t_1 \} }   \right) \vspace{0.2cm} 
		\end{array}
	\end{equation}
	$\mathbb{P}$-a.s. on $\{\tau_{n+1} \geq  t_1 \}$.  Note that on the set  $\{ \tau_{n+1} > t_1 \}$, $\tau_{n+1}$ takes only  $n$ values, then by inductional  hypothesis and the Markov property of the process $\{X_t^{\pi}\}$, we have that 
	\begin{equation}
		\begin{array}{l}
			\mathbb{E}_x \left(  \left[  \displaystyle \int_{t_1}^{\tau_{n+1} \wedge \tau^\pi} \phi( X_s^\pi, s ) ds  + V\left(X^{\pi}_{\tau_{n+1} \wedge \tau^\pi} \right) \right]  1_{ \{\tau_{n+1}  > t_1 \}} 1_{ \{\tau^\pi  \geq  t_1 \} }     \right) \vspace{0.2cm}\\
			\leq  \mathbb{E}_x\left( V(X_{t_1}^\pi)   1_{ \{\tau_{n+1}  > t_1 \}} 1_{ \{\tau^\pi \geq t_1 \} }  \right)	
		\end{array}
	\end{equation}
	Utilizing this  inequality into (\ref{dpp:eqn7}) we obtain
	\begin{equation}\label{dpp:eqn8}
		\begin{array}{l}
			\mathbb{E}_x \left[  \displaystyle \int_{0}^{\tau_{n+1} \wedge \tau^\pi} \phi( X_s^\pi, s ) ds  + V(X^{\pi}_{\tau_{n+1} \wedge \tau^\pi}) \right] \vspace{0.2cm} \\
			\leq 	\mathbb{E}_x \left[ 1_{ \{\tau^\pi < t_1 \} }  \displaystyle \int_{0}^{\tau^\pi} \phi( X_s^\pi, s ) ds    \right] \vspace{0.2cm} \\
			+  \; 	\mathbb{E}_x \left(  \left[  \displaystyle \int_{0}^{t_1} \phi( X_s^\pi, s ) ds  + V(X^{\pi}_{t_1}) \right]  1_{ \{\tau_{n+1}  > t_1 \}} 1_{ \{\tau^\pi \geq t_1 \} }     \right. \vspace{0.2cm} \\
			+  \; 	\left.  \left[  \displaystyle \int_{0}^{t_1} \phi( X_s^\pi, s ) ds  + V(X^{\pi}_{t_1}) \right]  1_{ \{\tau_{n+1}  = t_1 \}} 1_{ \{\tau^\pi  \geq  t_1 \} }     \right) \vspace{0.2cm} \\
			\leq  \;  \mathbb{E}_x  \left[ 1_{ \{\tau^\pi < t_1 \} }  \displaystyle \int_{0}^{\tau^\pi} \phi( X_s^\pi, s ) ds    \right]  \vspace{0.2cm} \\
			+  \; \mathbb{E}_x \left[   1_{ \{\tau^\pi \geq t_1 \} } \left( V(X_{t_1}^\pi)  +  \displaystyle \int_{0}^{t_1} \phi( X_s^\pi, s ) ds \right)   \right]  \vspace{0.2cm} \\
			=  \; \mathbb{E}_x \left[  \displaystyle \int_{0}^{t_1  \wedge \tau^\pi} \phi( X_s^\pi, s ) ds  + V(X_{t_1}^\pi)    \right] \leq V(x)
		\end{array}
	\end{equation}
	$\mathbb{P}$-a.s. on $\{\tau_{n+1} \geq  t_1 \}$. The las inequality  is due to (\ref{dpptheo}) for fixed time $t_1$. Consequently, we obtain   $A(x,\tau_n) \leq V(x) $ for all $n$, whence  $A(x,\tau_n) = V(x)$ for all $n$. Finally, utilizing dominated convergence theorem, together with the continuity of the value function, we proof the general identity  of (\ref{dpptheo}).  This complete the proof. 
\end{proof}

\section{The Hamilton-Jacobi-Bellman equation}\label{Hamilton-Jacobi}
We are now ready to calculate  the  Hamilton-Jacobi-Bellman (HJB) equation associated to our optimization problem  (\ref{optimizationprob}).   Given   $\varepsilon >0$ be a constant there exist a strategy $\pi \in \prod^{ad}$ and  $h>0$   such that  $V(h,x) < V^{\pi}(h,x) + \varepsilon$.  Using  It\^{o}'s Lemma for a function of two variables  and  a similar argument as  in Section 2.2 of Azcue and Muler \cite{AzMuBook},  we find 
\begin{equation}\label{generator}
	\begin{array}{ccl}
		\lim_{t\downarrow 0}\displaystyle\frac{1}{t} \mathbb{E}_x \left[  f(\tau^\pi\wedge t, X_{\tau^\pi\wedge t}^\pi) - f(0,x) \right] &=& \displaystyle\frac{1}{2} \left[ \sigma^\pi(x) \right]^2  f_{xx}(0,x) +  \mu^\pi(x) f_{x}(0,x) + f_t(0,x)  \vspace{0.3cm}\\
		&&  + \lambda \displaystyle\int_{[0, \infty)} [ f(0, x-z) - f(0, x) ]dF_U(z)
	\end{array}
\end{equation}
where $x\in\R$,  $\mu^\pi(x)=  c + \mu \pi(x) + r( x   - \pi(x)) $,   $\sigma^{\pi}(x)=\sigma \pi(x)$  and  $F_U$ is the distribution function of claim size. By Theorem \ref{DPP} we have that 
\begin{equation}
	\begin{array}{ccl}
		V(0,x) &=& \sup_{\pi} V^\pi(0,x) \vspace{0.2cm} \\
		& \geq &  \mathbb{E}_x \left[ \displaystyle\int_0^{\tau^\pi\wedge h} \phi(s, X_s^\pi)ds  +  V^\pi(\tau^\pi\wedge h, X_{\tau^\pi\wedge h}^\pi )  \right] \vspace{0.2cm} \\
		& > &  \mathbb{E}_x \left[ \displaystyle\int_0^{\tau^\pi\wedge h} \phi(s, X_s^\pi)ds  +  V (\tau^\pi\wedge h, X_{\tau^\pi\wedge h}^\pi )  \right] - \varepsilon \vspace{0.2cm} \\
	\end{array}
\end{equation}
In the before equation we can let $\varepsilon$ tend to zero. Assuming  that $V(t,x)$ is twice continuously differentiable in $x$ and continuously differentiable in $t$, and dividing by $h$ and letting  $h\downarrow 0$ we get 
$$\phi(0,x) + \displaystyle\frac{1}{2} \left[ \sigma^\pi(x) \right]^2  f_{xx}(0,x) +  \mu^\pi(x) f_{x}(0,x) + f_t(0,x)   + \lambda \mathcal{I}\left( V(0,x) \right)  \leq 0 $$
where  $ \mathcal{I}[V] \left( 0,x \right) = \displaystyle\int_{[0, \infty)} [ V(0, x-z) - V(0, x) ]dF_U(z)$. This inequality has to hold for all $\pi$. Therefore, this motives the Hamilton-Jacobi-Bellman equation by our optimization problem 
\begin{equation}\label{HJB}
	\phi(t,x) +  V_t(t,x)   + \mathcal{L}[V](t,x) = 0, \;\; (t,x)\in [0,T)\times [0, \infty) 
\end{equation}
with the boundary conditions $V(t,0)=V(T,x)=0$, and 
where  $\mathcal{L}$ is the following second order partial integro-differential operator
$$\mathcal{L}[V](t,x) =  \displaystyle\sup_{\pi}	\left\{   \displaystyle\frac{1}{2} \left[ \sigma^\pi(x) \right]^2  V_{xx}(t,x) +  \mu^\pi(x) V_{x}(t,x) + \lambda \mathcal{I}[V] \left( t,x \right)   \right\}$$
for $V\in \mathcal{C}_0^{1,2}\left( [0,T]\times[0,\infty) \right)$, the set of all functions  twice continuously differentiable in $x$ and continuously differentiable in $t$.
\medskip

In addition, by Lemma \ref{lemmaconcave},  $V:[0,T]\times \mathbb{R}_+ \to \mathbb{R}_+$ is strictly increasing and concave in $x$ with the boundary condition $V(t,0)=V(T,x)=0$.  Because the left hand side of (\ref{HJB})  is quadratic in  $\pi$, we find that

\begin{equation}\label{optimal_strategy1}
	\pi^* = - \dfrac{(\mu - r) V_x}{\sigma^2 V_{xx}}
\end{equation}
Strategy  $\pi^*$ will be our candidate for the optimal strategy. As $\pi(X_t^\pi)$ is the amount invested in the risky asset, we have that   $\pi(X_t^\pi)= \theta(t,X_t^\pi) X_t^\pi$, where $\theta$ represent the proportion of the surplus invested  in the risky asset at time $t$ (hence $\theta(t,X_t^\pi) \in[0,1]$ for all $t\in [0,T]$). Then, using (\ref{optimal_strategy1}) we find that  
\begin{equation}\label{optimal_strategy2}
	\theta^*(t,x) = - \dfrac{(\mu - r) V_x}{\sigma^2 xV_{xx}}
\end{equation}
Note that, if the expected return  of the stock market $\mu$ is greater than  the risk-free rate $r$, then   $\theta^*(t,x) \geq 0$, but that does not guarantee that $\theta^*(t,x)$ is less than 1.  Thus, if $\mu > r$ the right strategy would be  $\min\{\theta^*(t,x), 1\}$. In the other hand, if $\mu \leq r$ the right strategy will be  $\theta^*(t,x) \equiv 0$.

\begin{thm}\label{existencia}
	Suppose that there exists  a solution $f(t,x)$ of (\ref{HJB}) that is a function  twice continuously differentiable in $x$ and continuously differentiable in $t$ with boundary conditions $f(T,x)=0$. Then  $V(t,x)\leq f(t,x)$.  If 
	\begin{equation}
		\theta^*(t,x) = - \dfrac{(\mu - r) f_x}{\sigma^2 x f_{xx}}
	\end{equation}
	is bounded and $f(t,0)=0$ for all $t$, then  $V(t,x)=f(t,x)$ and an optimal strategy is given by $\{ \pi^*(t,X_t^{\pi^*}) : t\in [0, \tau^{\pi^*}\wedge T] \}$.
\end{thm}

\begin{proof}
	Let  $\pi\in\Pi^{ad}$ and  let $0< T_1 < \cdots < T_i < \cdots < t$ be the   interarrival times in $[0, t)$.  The controlled investment  process $X_t^{\pi}$ has  finite jumps on each finite time interval $[0, t)$. The jump size of the process $X_t^{\pi}$ at time $t$ es denoted by  $\Delta X_t^{\pi} = X^{\pi}_t - X^{\pi}_{t^-}$, then 
	$$\Delta X_{T_i}^{\pi} = X^{\pi}_{T_i} - X^{\pi}_{{T_i}^-} = - U_i$$
	where  $U_i$ is the claim size at time $T_i$.  Consider  the process $f(t,X_t^\pi)$ conditioned on  $X_s=x$.  Using  It\^{o}'s Lemma for a jump process (see for details  \linebreak \cite{Tankov2004}),  we have  that 
	$$
	\begin{array}{c}
		f(t,X_t^{\pi}) - f(s,x) =  \displaystyle\int_s^t \left[  \dfrac{\partial f}{dt}(s, X_s^{\pi}) +  \mu^{\pi}\dfrac{\partial f}{dx}(s, X_{s^-}^{\pi}) +  \dfrac{[\sigma^\pi]^2}{2}\dfrac{\partial^2 f}{d^2x}(s, X_{s^-}^{\pi})   \right]ds \vspace{0.2cm}\\
		+  \displaystyle\int_s^t \sigma^\pi \dfrac{\partial f}{dx}(s, X_{s^-}^{\pi})dW_s + \sum_{ \{  i\geq 1 : \;  s \leq  T_i < t \} } \left[ f(T_i, X_{T_i^{-}}^{\pi} + \Delta X_{T_i}^\pi ) - f(T_i, X_{T_i^{-}}^{\pi})  \right]
	\end{array}
	$$
	It is known that the stochastic integral is a martingale. Thus  
	$$ \mathbb{E}_x\left[  \displaystyle\int_s^t \sigma^\pi \dfrac{\partial f}{dx}(s, X_{s^-}^{\pi})dW_s  \right]=0.$$ 
	Then, by   the Hamilton-Jacobi-Bellman equation (\ref{HJB})  we have that 
	$$\mathbb{E}_x\left[ f(\tau^\pi\wedge T ,X_t^{\pi})  + \int_{s}^{\tau^\pi\wedge T} \phi(r ,X_r^{\pi})dr  \right]  \leq f(s,x) $$
	Because   $ f(\tau^\pi\wedge T ,X_t^{\pi}) \geq 0$, we obtain the following inequality
	$$V^{\pi}(s,x)=\mathbb{E}_x\left[ \int_{s}^{\tau^\pi\wedge T} \phi(r ,X_r^{\pi})dr \right]  \leq f(s,x) $$
	Now, if we take supremum over all strategies $\pi\in\Pi_{ad}$, we have that $V(s,x)\leq f(s,x)$.   
	\medskip
	
	Suppose now that  $\theta^*$ is bounded by a value $\overline{\theta}=\sup \theta(t,x)$. Note that  by the Hamilton-Jacobi-Bellman equation (\ref{HJB}),   $\theta^*>0$ whenever  $X_t^*>0$, where $\{ X_t^*\}$ is the process with optimal  investment strategy $\pi^*$. Choose $\varepsilon < x$ and define $\tau_{\varepsilon} = \inf\{t>s : X_t^* < \varepsilon \}$. Note that  $\tau_{\varepsilon} \to \tau^*$ converges monotonically as $\varepsilon \to 0$, where $\tau^*$ is the  time to ruin of the  process $X_t^*$.  
	\medskip
	
	As $f$ is a  concave function, we have that   $f_{xx}<0$. Thus,  $f_x$ is increasing on $x$  and $f_x(t,x)\leq f_x(t,\varepsilon)$ for all $t$.  Then, $f_x$ is bounded on  $[0,T]\times [\varepsilon, \infty \rangle$. Now, by the Hamilton-Jacobi-Bellman equation (\ref{HJB}) and the It\^{o} formula we obtain 
	\begin{equation}
		\mathbb{E}_x\left[ f(t ,X_t^{*})  + \int_{s}^{t} \phi(r , X_r^{*})dr  \right]  = f(s,x)  +  \mathbb{E}_x\left[  \displaystyle\int_s^t \sigma^{\pi^*} \dfrac{\partial f}{dx}(r, X_r^{*})dW_r  \right]
	\end{equation}
	In addition, as $\theta^*$ is bounded, then the second moment of   $X_t^{*}$ is bounded for all $t$, and  $\int_{0}^{T}\mathbb{E}[X_t^{*}]dt < \infty$. Thus, we have that    $\displaystyle\int_s^t \sigma^{\pi^*} \dfrac{\partial f}{dx}(r, X_r^{*})dW_r$ is a martingale (for more details see  \cite{Karatzas1991},  \cite{Watanabe1981}). Therefore, we obtain the following equation,
	$$	\mathbb{E}_x\left[ f(\tau_{\varepsilon} \wedge T ,X_{\tau_{\varepsilon} \wedge T}^{*})  + \int_{s}^{\tau_{\varepsilon} \wedge T} \phi(r , X_r^{*})dr  \right]  = f(s,x) $$
	\begin{figure}[t] 
		\centering
		\includegraphics[width=9cm]{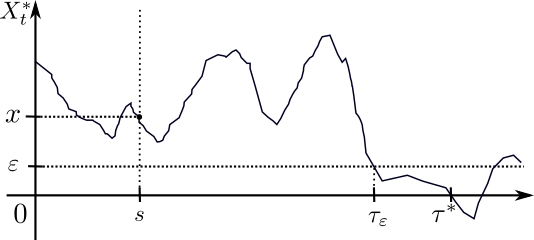}
		\caption{Graphical representation of the stopping time $\tau_{\varepsilon}$.}
	\end{figure}
	Notice  that,  as $\varepsilon \to 0$  we have that  
	$$\int_{s}^{\tau_{\varepsilon} \wedge T} \phi(r , X_r^{*})dr  \to \int_{s}^{\tau^* \wedge T} \phi(r , X_r^{*})dr$$ 
	monotonically. The first term can be written as follow    
	$$\mathbb{E}_x\left[ f(\tau_{\varepsilon},\varepsilon ) :  \tau_{\varepsilon} \leq T \right] +  \mathbb{E}_x\left[ f( T, X_T^{*}) :  \tau_{\varepsilon}>T \right] $$
	The term $\mathbb{E}_x\left[ f(\tau_{\varepsilon},\varepsilon ) :  \tau_{\varepsilon} \leq T \right]$ is bounded by  $f(0,\varepsilon)\mathbb{P}[\tau_{\varepsilon} \leq T]$, which is uniformly bounded. Thus, by the boundary condition,  it converges to $\mathbb{E}_x\left[ f(\tau^*, 0) :  \tau^* \leq T \right]=0$. Again, by boundary condition,  the second term converges monotonically  to  $$\mathbb{E}_x\left[ f( T, X_T^{*}) :  \tau^* > T \right]=0.$$ This completes  the proof.  
\end{proof}
\medskip

Finally, suppose that $f(t,x)$ is a solution of (\ref{HJB})  twice continuously differentiable in $x$ and continuously differentiable in $t$ with boundary conditions $f(T,x)=0$. Assume that there exists  some  constant  $M<\infty$ such that  
$$\theta^*(t,x) = - \dfrac{(\mu - r) f_x}{\sigma^2 x f_{xx}} \leq M$$
Remember that  $f_x > 0$ and  $f_{xx}<0$. Then 
$$\dfrac{\partial }{\partial x} \ln \left( f_x(t,x) \right) \leq -\dfrac{(\mu - r)}{M\sigma^2}\times \dfrac{1}{x}$$
Therefore, for some constants $\alpha\in \langle 0,1\rangle$  and $K>0$ we have that 
$$f(t,x) \leq K x^{1-\alpha}$$
This bound for the solution motivate trying a solution of the form  $V(t,x)=f^\alpha(t)x^{1-\alpha}$.

\section{Numerical example}\label{numerical}
We apply the method  to the  Cobb-Douglas type utility function   $\phi(t,x)=x^{1-\alpha} e^{-\kappa\alpha t}$, where  $\alpha\in \langle 0,1 \rangle$, $\kappa\in \mathbb{R}$, in order to  solve the optimization problem  (\ref{optimizationprob}) trying a solution of the form  $V(t,x)=f^\alpha(t)x^{1-\alpha}$ for a continuously differentiable function  $f: [0,T]\to \mathbb{R}_+$ with boundary condition $f(T)=0$.  By (\ref{optimal_strategy1}) we have that  the optimal strategy for  our problem  is given by
\begin{equation}\label{proportion1}
	\pi^*(x) = \dfrac{(\mu - r)}{\sigma^2\alpha}x
\end{equation}
and  the optimal  proportion of the surplus invested  in the risky asset at time $t$ is given by
\begin{equation}\label{proportion2}
	\theta^*= \dfrac{(\mu - r)}{\sigma^2\alpha}
\end{equation}
Notice that in this case solving our optimization problem  we obtain  that  the optimal proportion to be invested in the  risky asset is the Merton ratio, introduced  in \cite{Merton1969}.  Next,  by the HJB equation (\ref{HJB}) we obtain the following equation
$$
\begin{array}{c}
	\phi(x,t) +  \alpha f^{\alpha-1}(t) f'(t)x^{1-\alpha}   +  \left[ c+rx + (\mu-r)\theta^* x \right] (1-\alpha) f^{\alpha}(t) x^{-\alpha} \vspace{0.2cm} \\
	- \dfrac{1}{2}\sigma^2 (\theta^*)^2 x^2\alpha(1-\alpha) f^\alpha(t)x^{-\alpha-1} +   \lambda f^\alpha(t) \mathbb{E}\left( (x-U)^{1-\alpha} -x^{1-\alpha}  \right) = 0 \\
\end{array}
$$
Thus, the function $f(t)$  satisfies the following relation
$$\alpha \dfrac{f'(t)}{f(t)} + \dfrac{e^{-\kappa \alpha t}}{f^\alpha(t)} -  \dfrac{1}{2}  \alpha(1-\alpha) \sigma^2 (\theta^*)^2  = K(x)$$
with boundary condition $f(T)=0$,   where 
$$K(x) = \dfrac{\lambda} {x^{1-\alpha}} \mathbb{E}\left( x^{1-\alpha} - (x-U)^{1-\alpha}  \right) - \dfrac{c(1-\alpha)}{x} -  \left[r + (\mu-r)\theta^* \right](1-\alpha)  $$ 
for all $t$ and $x$. 
Thus, we have that $f$ satisfies the differential equation
$$\alpha f'(t)  + f^{1-\alpha}(t) e^{-\kappa\alpha t} = \left( K + \dfrac{1}{2}  \alpha(1-\alpha) \sigma^2 (\theta^*)^2 \right)f(t) $$
In order to solve these  we are now going to use the substitution $z(t)=f^\alpha(t)$ obtaining the differential equation $z'(t) -Rz(t)= e^{-\kappa\alpha t}$, where $R= K + \dfrac{1}{2}  \alpha(1-\alpha) \sigma^2 (\theta^*)^2$.  
We thus find that the solution is given by
\begin{equation}
	f(t)=\dfrac{e^{-\kappa t }}{(R+\kappa\alpha)^{1/\alpha}} \left( e^{-(R+\kappa\alpha)(T-t) }  -  1  \right)
\end{equation}
\medskip

\noindent {\bf Exponential claim sizes:}
Assume  that  the claim sizes $U_i$ has $Exp(\theta)$ distribution, i.e.,  $U_i \sim Exp(\theta)$. Then, we obtain the following  identity
$$\mathbb{E} \left( x- U \right)^{1-\alpha}  = B(2-\alpha ,1) x^{2-\alpha} {}_{1}F_1(1,3-\alpha,-\theta x) \leq  B(2-\alpha,1) x^{2-\alpha} e^{-\theta x}$$
Thus,  we have that $\mathbb{E} \left( x- U \right)^\alpha  \to 0$  as $x\to \infty$. Using this result  is straightforward  prove  the  following asymptotic result for the function $K$.
\begin{prop}
	Let $\{U_i\}$ be a sequence of independent, identically distributed (i.i.d). random variables having $Exp(\theta)$ distribution. Then the following hold.
	\begin{enumerate}
		\item If $x\to \infty$, then  
		$$K(x)\to  \lambda  -   \alpha(r+\theta^*(\mu-r) ) $$
		where  $\theta^*$ is the optimal  proportion of the surplus invested  in the risky asset.
		\item If $x\to 0^+$, then  $K(x)\to -\infty$.		  
	\end{enumerate}
\end{prop}

\subsection{Simulation of the ruin probability with and without optimal investment}


In this  section, we conduct Monte Carlo simulation  in order to compute and compare the behavior of the ruin probability with and without optimal investment. For the purpose of comparison, ruin probability with and without optimal investment are calculated to Pareto and Weibull distribution as claim sizes.  The ruin probability  $\psi^{\pi}(x,T)$  for the  risk process with  investment  $X_t^{\pi}$ can be simulated by randomly drawing sample paths  according  to the process $X_t^{\pi}$ and  counting the trajectories that lead to ruin and dividing this number by the total number  $N$ of simulated trajectories.  Thus,  we get an unbiased estimator of the ruin probability 
\begin{equation}\label{key}
	\psi^{\pi}(x,T)=\mathbb{P}_x[\tau^\pi < T]	\approx	\hat{\psi}^{\pi}(x,T) = \dfrac{1}{N}  \sum_{i=1}^{N} 1_A(w_i)
\end{equation}
where $A$ is the set of all trajectories $w_i$ that lead to ruin up to time  $T$.  Notice that to estimate the ruin probability it is necessary to simulate a sample trajectories of the  jump-diffusion process 
$$dX_t^{\pi}= \left[c + \mu\pi_tX_t^{\pi} + r(1-\pi_t)X_t^{\pi} \right] dt +  \sigma\pi_t X_t^{\pi}dW_t - dQ_t $$
where $Q_t$ is a compound Poisson process. 
\medskip  

Now, we wish to simulate  paths of the process  $\{X_t^{\pi}\}$ without knowing  its distribution or an explicit solution to the equation SDE  in order to know if each simulated trajectory goes to ruin or not.    We can simulate a discretized version of the   jump-diffusion process.  In particular, we simulate a discretized trajectories, $\{ \hat{X}_0, \hat{X}_{h}, \hat{X}_{2h}, ... , \hat{X}_{nh}  \}$ where  $n$  is the number of time steps,  $h$ is a constant and $h=T/n$.  The smaller the value of $h$, the closer our discretized path will be to the continuous-time path of $X_t^{\pi}$ that we wish to simulate (for more details see example \cite{Karatzas1991}, \cite{Oksendal} and \cite{Watanabe1981}).  
\medskip

In the literature there are several discretization schemes available, the simplest approach is the Euler scheme. The Euler method is intuitive and easy to implement, and in our case we get the following discretize version
\begin{equation}\label{discrete}
	\hat{X}_{kh}^{\pi} = \hat{X}_{(k-1)h}^{\pi} + \left[c + \mu\pi \hat{X}_{(k-1)h}^{\pi} + r(1-\pi_t)\hat{X}_{(k-1)h}^{\pi} \right] \! h +  \sigma\pi \hat{X}_{(k-1)h}^{\pi}  Z_k  - \Delta Q_{kh}
\end{equation}
where the   $Z_k$  are   i.i.d. $N(0,h)$. 
\medskip 

\noindent Note that the risk process with investment  $X_t^{\pi}$ has  finite jumps on any finite interval $[0, T)$, and   in  absence of claims (or  between the jumps of $N_t$) is continuous and satisface the discretized version of the process
$$\hat{X}_{kh}^{\pi} = \hat{X}_{(k-1)h}^{\pi} + \left[c + \mu\pi \hat{X}_{(k-1)h}^{\pi} + r(1-\pi_t)\hat{X}_{(k-1)h}^{\pi} \right] h +  \sigma\pi \hat{X}_{(k-1)h}^{\pi}  Z_k$$
The jump size of the process $X_t^{\pi}$ at time $t$ is denoted by  $\Delta X_t^{\pi} = X^{\pi}_t - X^{\pi}_{t^-}$. The notation $X_{t^{-}}$  refers to $\lim_{s\to t^-}X_s$. Thus,  if the $n^{th}$ jump in the compound Poisson process occurs at time $t$ we have
$$X^{\pi}_{t} - X^{\pi}_{{t}^-} = - U_n$$
where  $U_n$ is the claim size at time $t=T_n$.  
\medskip

\noindent Taking in account previous remarks, an approach to stimulating a discretized version of $X_t^{\pi}$ on the interval $[0,T]$ is given by
\begin{enumerate}
	\item First simulate the arrival times in the compound Poisson process up to time  $T$.
	\item Use a pure diffusion discretization between the jump times.
	\item At the $n^{th}$  arrival time    $T_n$, simulate the $n^{th}$   claim size $U_n$ conditional on the value of the discretized process,  $\hat{X}_{T_n}^{\pi}$, immediately before $T_n$.
\end{enumerate}

\noindent We consider the case of a Weibull and Pareto distribution for  the claim size (severity) $U$ with $U\sim Weibull(1,50)$ and $U\sim Pareto(25,2)$,  frequency  parameter  $\lambda = 1$, premium rate  $c = 65$ and safety  loading factor $\rho=30\%$.  Parameters to  the optimal investment are: risk-free rate   $r=8.4\times 10^{-4}$, expected return  $\mu=10^{-3}$, volatility $\sigma^{2}=10^{-3}$ and risk aversion   $\alpha=0.2$. In our  simulation we have around of 0.48\% monthly excess return and Merton ratio  $\pi=0.8$. The choice of the parameters is purely academic and  $T=1$ year. 

\begin{figure}[h]
	\centering
	\caption{Ruin probability with claim size having Weibull and Pareto distribution: without investment (solid line) and with optimal investment (dashed line).}
	\label{ruin.graph}
	\medskip
	\includegraphics[width=13.3cm]{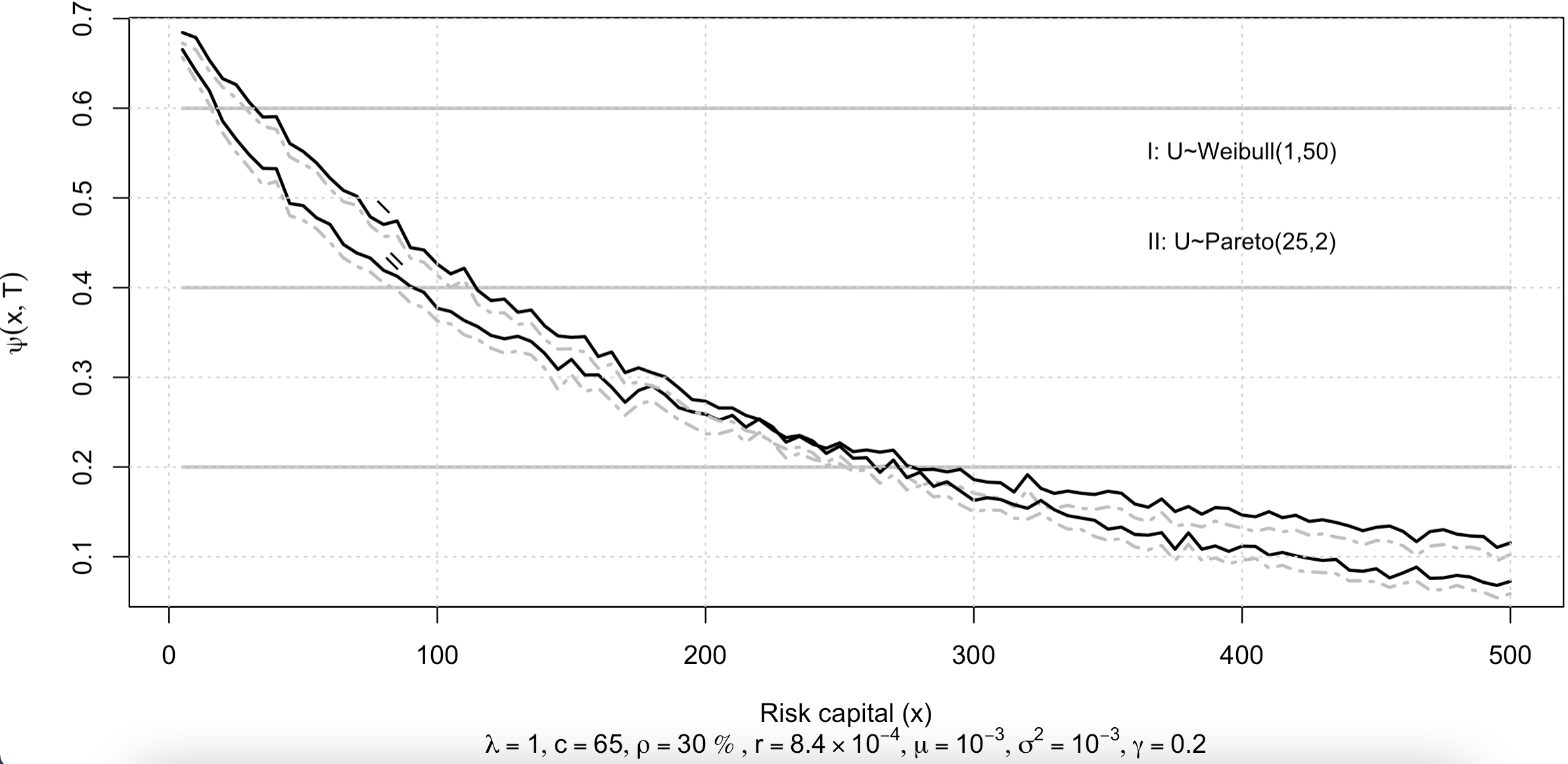}	
\end{figure}

\noindent Using the discretized version of the process  \ref{discrete},  we simulate  $n=10000$  trajectories for each case.  Numerical results,  in Figure  \ref{ruin.graph}, show that the  optimal strategy in one year reduce  the ruin probability of the portfolio for diferent claim size distributions. Solid line is the  ruin probability without investment and dashed line  is the ruin probability with investment. The difference in ruin probabilities seems to be small, but notice that, with Exponential claim size for capital risk $x=100$  we have a ruin probability $\hat{\psi}^{\pi}(x,T)$ of 0.4306 with optimal investment, while the ruin  probability is 0.4406 without investment $\hat{\psi}(x,T)$. In Table \ref{summary}, we summarize our results for Exponential, Pareto and Weibull distributions for  the claim size, and  different values of risk capital $x$.

\begin{table}
	\centering
	\caption{Ruin probability with claim size having Exponential, Pareto and  Weibull  distribution: without investment $\hat{\psi}(x,T)$ and with optimal investment $\hat{\psi}^{\pi}(x,T)$.}
	\label{summary}
	\medskip
	\begin{tabular}{|c|cc|cc|cc|}
		\hline 
		\multirow{2}{*}{$x$} & \multicolumn{2}{c|}{Exponential}      & \multicolumn{2}{c|}{Pareto}           & \multicolumn{2}{c|}{Weibull}         \\ \cline{2-7} 
		& \multicolumn{1}{c|}{$\hat{\psi}(x,T)$} & $\hat{\psi}^{\pi}(x,T)$  & \multicolumn{1}{c|}{$\hat{\psi}(x,T)$} & $\hat{\psi}^{\pi}(x,T)$   & \multicolumn{1}{c|}{$\hat{\psi}(x,T)$} & $\hat{\psi}^{\pi}(x,T)$  \\  \hline 
		100                & \multicolumn{1}{c|}{0.4406}  & 0.4372 & \multicolumn{1}{c|}{0.377}   & 0.3728 & \multicolumn{1}{c|}{0.4262}  & 0.424 \\ \hline
		200                & \multicolumn{1}{c|}{0.274}   & 0.2676 & \multicolumn{1}{c|}{0.2588}  & 0.247  & \multicolumn{1}{c|}{0.2734}  & 0.268 \\ \hline
		400                & \multicolumn{1}{c|}{0.1014}  & 0.097  & \multicolumn{1}{c|}{0.1464}  & 0.1418 & \multicolumn{1}{c|}{0.1118}  & 0.106 \\ \hline
	\end{tabular}
\end{table}

\subsection*{Acknowledgements.} 
JCH was supported by PROCIENCIA-CONCYTEC, contract  427-2019.  AR was supported by ANID-Chile (Fondecyt grant1231188)




\begin{thebibliography}{10}

\bibitem[Andersen and Sparre (1957)]{Sparre1957}
\textsc{Andersen, E. Sparre}. (1957). 
On the collective theory of risk in case of contagion between the claims. Transactions XVth International Congress of Actuaries, New York, II, 219-229. 

\bibitem[Asmusen (1984)]{Asmusen1984}
\textsc{Asmusen,  S.}
\textit{Aproximations for the probability of ruin within finite time},
Scand. Act. J. 31-57 (1984).

\bibitem[Asmusen (1989)]{Asmussen1989}
\textsc{Asmussen, S.} 
Risk theory in a markovian environment. 
Scand. Actuarial I., 66 - 100. (1989) 

\bibitem[Asmusen and Albrecher (2010)]{Asmusen_ruin2010}
\textsc{Asmusen,  S.} and \textsc{Albrecher,  H.}: 
\textit{Ruin probabilities},
World Scientific Publishing, Second Edition (2010)

\bibitem[Azcue and Muler (2014)]{AzMuBook}
\textsc{Azcue P.} and \textsc{Muler N.}.
{\it Stochastic optimization in insurance: A dynamic programming approach}
Springer (2014).

\bibitem[Badaoui at al. (2018)]{Badaoui2018}
Badaoui, M.;   Fernández, B. and  Swishchuk, A. 
{An Optimal Investment Strategy for Insurers in Incomplete Markets}. 
Risks,  2018: 6-31.

\bibitem[Browne  (1995)]{Browne1995}
\textsc{Browne S.} (1995). 
Optimal investment policies for a firm with a random risk process: exponential utility and minimizing the probability of ruin. 
\textit{Math. Oper. Res.}, 20(4), 937–958.

\bibitem[Cadenillas and  Zou (2014)]{Cadenillas2014}
Cadenillas, A. and  Zou, B., 
{Optimal investment and risk control policies for an insurer: Expected utility maximization},
Insurance: Mathematics and Economics, 2014),  Vol 58: 57-67.

\bibitem[Cont and Tankov (2004)]{Tankov2004} 
\textsc{Cont R. } and  \textsc{ Tankov  P. }. 
\textit{Financial Modelling with Jump Processes},
Chapman \& Hall / CRC Press (2004).

\bibitem[Cramer (1930)]{Cramer1930}
\textsc{Cram\'er, H.} (1930).
On the Mathematical Theory of Risk. Skandia Jubilee Volume, Stockholm. 

\bibitem[Cramer (1945)]{Cramer1945}
\textsc{Cram\'er, H.}: Collective Risk Theory. Skandia Jubilee Volume, Stockholm (1945)

\bibitem[De Finetti  (1957)]{De Finetti1957}
\textsc{De Finetti, B.} (1957). 
Su un’ impostazione alternativa dell teoria collettiva del risichio. 
\textit{Transactions of the XVth congress of actuaries}, (II), 433–443.

\bibitem[Diko et al. (2011)]{Diko} 
\textsc{Diko, P.}  and \textsc{Usabel, M.}: A numerical method for the expected penalty reward function in a Markov-modulated jump diffusion process. Insurance: Mathematical and Economics, V49, Issue 1, 126--131 (2011).

\bibitem[Embrechts and Veraverbeke (1982)]{Embrechts1982}
\textsc{Embrechts, P.} and  \textsc{Veraverbeke, N.}  
Estimates for the probability of ruin with special emphasis on the possibility of large claims. 
Insurance: Mathematics and Economics 1, 55 - 72. (1982)

\bibitem[Embrechts et al. (1997)]{Embrechts1997}
\textsc{Embrechts P.}, \textsc{Kluppelberg C},  and  \textsc{Mikosch T.}
Modelling Extremal Events for Insurance and Finance.
Stochastic Modelling and Applied Probability. Springer (1997).

\bibitem[Gordon (1959)]{Gordon1959}
\textsc{Gordon,M.J.} (1959).
Dividends, earnings and stock prices,
\textit{Review of Economics and Statistics}, 41, 99–105.

\bibitem[Hipp and Plum (2000)]{HP1} 
\textsc{Hipp, C.}   and  \textsc{Plum, M.} (2000)
Optimal investment for insurers.
Insurance: Mathematical and Economics, 27,  215--228.

\bibitem[Hipp and Plum (2003)]{HP2} 
\textsc{Hipp, C.}   and  \textsc{Plum, M.} (2003).  
Optimal investment for investors with state dependent income, and for insurers.   
Finance and Stochastics, 7,  299--321 (2003).

\bibitem[Hipp and Vogt (2003)]{HipVog}
\textsc{Hipp C.} and  \textsc{Vogt M.} 
{\it Optimal dynamic XL reinsurance}. ASTIN Bull 33(2):193–207  (2003).

\bibitem[Hipp (2004)]{Hipp2004} 
\textsc{Hipp, C.}:  Stochastic control with Application in Insurance.  Stochastics Methods in Finance,  pp. 127--164. Lecture notes in Mathematics, No 1856, Springer, Berlin (2004).

\bibitem[Ikeda and Watanabe, 1981]{Watanabe1981}
\textsc{Ikeda N.} and \textsc{Watanabe, S.} (1981).
\textit{Stochastic Differential Equations and Diffusion processes}.
North Holland, Amsterdam.

\bibitem[Karatzas and Shreve, 1991]{Karatzas1991}
\textsc{Karatzas I.} and \textsc{Shreve, S.} (1991).
\textit{Brawnian Motion and  Stochastic Calculus}.
Springer.

\bibitem[Korn and   Wiese (2008)]{Korn2008}
Korn, R.  and   Wiese, A.,
{Optimal Investment and Bounded Ruin Probability: Constant Portfolio Strategies and Mean-variance Analysis}, 
ASTIN Bulletin,  2008,  38(2):423-440.

\bibitem[Lundberg (1903)]{Lundberg1903}
\textsc{Lundberg, F.} (1903).
I. Approximerad Framst\"{a}llning av Sannolikhets - funktionen. 
Almqvist-Wiksell, Uppsala.

\bibitem[Lundberg (1909)]{Lundberg1909}
\textsc{Lundberg, F.} (1909).
\"{U}ber Die Theorie Der R\"{u}ckversicherung. Transactions of the VIth International Congress of Actuaries, vol. 1, pp. 877-948 

\bibitem[Landsman  and Sherris (2001)]{Landsman2001}
\textsc{Landsman, Z.} and  \textsc{ Sherris, M.}: 
Risk measures and insurance premium principles. 
Insurance: Mathematics and Economics, 29(1), 103 - 115. (2001)

\bibitem[Merton R. (1969)]{Merton1969} 
{\sc Merton Robert C. }(1969).
Lifetime Portfolio Selection under Uncertainty: The Continuous-Time Case. 
{\it  The Review of Economics and Statistics}. 
\textbf{51}, No. 3 , pp. 247--257

\bibitem[Merton R. (1971)]{Merton1971} 
{\sc Merton Robert C. }(1971).
Optimum Consumption and Portfolio Rules in a Continuous-Time Model.
{\it  Journal of Economic  Theory}. 
\textbf{3}, pp. 373--413

\bibitem[Mikosh  (2004)]{Mikosh2004}
\textsc{Mikosh, T.} 
Non-Life Insurance Mathematics: an introduction with Stochastic Processes.
Springer (India), New Delhi  (2004)

\bibitem[Modigliani1 and Miller (1961)]{Modigliani1961}
\textsc{Miller, M. H.}  and  \textsc{Modigliani, F.} (1961). 
Dividend policy, growth, and the valuation of shares, 
\textit{The Journal of Business}, 34(4), 411–433.

\bibitem[Ramasubramanian (2009)]{R} 
\textsc{Ramasubramanian, S.}:  Lectures on Insurance Models, Hindustan Book Agency, India, New Delhi, (2009).

\bibitem[Rolski et al. (2005)]{RSST} 
\textsc{Rolski, T.}  ,  \textsc{Schmidli, H.},   \textsc{Schmidt, V.}  and  \textsc{Teugels,  J.L.}:  Stochastic Processes for Insurance and Finance. 8th edition. Elsevier India, New Delhi, (2005).

\bibitem[Oksendal (2000)]{Oksendal} 
\textsc{Oksendal Bernt}.  Stochastic Differential Equations, An Introduction with Applications. 
Springer, (2000).

\bibitem[Schmidli  (2008)]{SchmBook}
\textsc{Schmidli  H.}  
{\it Stochastic control in insurance}. Springer, New York (2008).

\bibitem[Teugels  and   Sundt (2004)]{TS} 
\textsc{Teugels, J.L.}  and  \textsc{Sundt B.} : Encyclopedia of Actuarial Science, 3 Vols. Wiley, Chichester (2004).

\bibitem[Thorin (1971)]{Thorin1971} 
\textsc{Thorin, O.} 
Further remarks on the ruin problem in case the epochs of the claims form a renewal process. Skand. AktuarTidskr., 14 - 38 and 121 - 142. (1971) 

\bibitem[Thorin (1974)]{Thorin1974} 
\textsc{Thorin, O.}
On the asymptotic behavior of the ruin probability for an infinite period when the epochs of claims form a renewal process. 
Scand. Actuarial J., 81 - 99.  (1974) 

\bibitem[Thorin (1975)]{Thorin1975}
\textsc{Thorin, O.}  
Stationarity aspects of the Sparre Andersen risk process and the corresponding ruin probabilities. 
Scand. Actuarial J., 87 - 98. (1975)

\bibitem[Wang and Dhaene (1998)]{Wang98}
\textsc{Wang, S.} and  \textsc{Dhaene, J.}:  
Comonotonicity, correlation order and premium principles. 
Insurance: Mathematics and Economics, 22(3), 235 - 242 (1998)

\bibitem[Wang et al. (1997)]{Wang97}
\textsc{Wang, S.},  \textsc{Young, V.R.}  and \textsc{ Panjer, H.}:  Axiomatic characterization of insurance prices. Insurance: Mathematics and Economics, 21(2), 173 - 183 (1997)

\bibitem[Young  (2006)]{Young2006}
\textsc{Young, V. R.} 
Premium Principles. 
John Wiley and Sons, Ltd. (2006)

\bibitem[Artzner at al. (2023)]{Artzner2023}
Artzner, P.;   Eisele, K. T. and   Schmidt, T. 
{Insurance-finance arbitrage}. 
Mathematical Finance, 2023:  00, 1-35.

\bibitem[Robben et al. (2022)]{Robben2022}
Robben, J.; Katrien, A. and Sander, D.  
{Assessing the impact of the covid-19 shock on a stochastic multi-population mortality model}. 
Risks. 2022:  10(2): 26.

\bibitem[Dhaene at al. (2013)]{Dhaene2013}
Dhaene, J.; Kukush, A.; Luciano, E.; Schoutens, W. and  Stassen, B.
{A note on the (in-)dependence between financial and actuarial risks.} Insurance: Mathematics and Economics, 2013:  52, 522–531.

\bibitem[Ferguson (1965)]{Ferguson1965}
Ferguson, Thomas S. 
{Betting systems which minimize the probability of ruin}. 
Journal of the Society for Industrial and Applied Mathematics, 1965: 13: 795-818.

\bibitem[Zariphopoulou (2001)]{Zariphopoulou2001}
Zariphopoulou, T.  
{A solution approach to valuation with unhedgeable risks}. 
Finance Stochast 2001: 5: 61–82.

\bibitem[Badaoui and  Fernández (2013)]{BF2013}
Badaoui, M.  and  Fernández,  B. 2013. 
{An optimal investment strategy with maximal risk aversion and its ruin probability in the presence of stochastic volatility on investments}. 
Insurance: Mathematics and Economics 53: 1-13.

\end{thebibliography}
\end{document}